\documentclass[final,journal,letterpaper]{IEEEtran}
\usepackage{latexsym}
\usepackage{amsmath}
\usepackage{amssymb}
\usepackage{enumerate}
\usepackage{cite}
\newcommand{\vc}{{\sf c}}
\newcommand{\vu}{{\sf u}}
\newcommand{\vv}{{\sf v}}

\def\Z{\mathbb{Z}}

\newtheorem{thm}{Theorem}[section]

\newtheorem{example}[thm]{Example}

\date{}
\begin{document}
\title{New Constant-Weight Codes from \\Propagation Rules}
\author{Yeow~Meng~Chee,~\IEEEmembership{Senior Member,~IEEE,}
~Chaoping~Xing~and~Sze~Ling~Yeo
\thanks{ 
The research of Y. M. Chee is supported in part by the National
Research Foundation of Singapore under Research Grant NRF-CRP2-2007-03
and by the Nanyang Technological University
under Research Grant M58110040.
The research of C. Xing is supported in part by the National
Research Foundation of Singapore under Research Grant NRF-CRP2-2007-03
and the Singapore Ministry of Education under Research Grant T208B2206.}
\thanks{Y. M. Chee, C. Xing and S. L. Yeo
are with Division of Mathematical Sciences, School of Physical and
Mathematical Sciences, Nanyang Technological University, 
21 Nanyang Link, Singapore 637371 (email: {\tt ymchee@ntu.edu.sg},
{\tt xingcp@ntu.edu.sg}, {\tt yeosl@ntu.edu.sg}).}
\thanks{Corresponding author: C. Xing.}
}

\maketitle
\begin{abstract}
This paper proposes some simple propagation rules which give rise to new binary
constant-weight codes.

\end{abstract}


\section{Introduction}

\PARstart{T}{he} ring $\Z/q\Z$ is denoted $\Z_q$.
We endow $\Z_q^n$ with
the {\em Hamming distance} metric $\Delta$: for $\vu,\vv\in \Z_q^n$,
$\Delta(\vu,\vv)$ is the number
of positions where $\vu$ and $\vv$ differ.
A ($q$-ary) {\em code} of length $n$ is a subset ${\cal C}\subseteq \Z_q^n$.
The elements of $\cal C$ are called {\em codewords}, and the {\em size}
of $\cal C$ is the number of codewords it contains. 
The {\em minimum distance} of a code ${\cal C}$
is $\Delta({\cal C})=\min_{\vu,\vv\in{\cal C},\vu\not=\vv} \Delta(\vu,\vv)$.
We often denote by $(n,d)_q$-code a $q$-ary code of length $n$ and
minimum distance at least $d$.

The {\em weight}, ${\rm wt}(\vu)$, of $\vu\in \Z_q^n$ is its distance from the origin,
that is, ${\rm wt}(\vu)=\Delta(\vu,{\sf 0})$.
For $0\leq w\leq n$, the ($q$-ary) {\em Johnson space} $J_q^n(w)$
is the set of all elements of $\Z_q^n$ having weight $w$, that is,
$J_q^n(w)=\{\vu\in \Z_q^n:{\rm wt}(\vu)=w\}$.
A ($q$-ary) {\em constant-weight code} of length $n$, distance $d$, and weight $w$,
denoted $(n,d,w)_q$-code, is a code ${\cal C}\subseteq J_q^n(w)$ such
that $\Delta({\cal C})\geq d$. 

We adopt the convention throughout
this paper that if $q$ is not specified, then we assume $q=2$.
Hence, for example, an $(n,d,w)$-code refers to an $(n,d,w)_2$-code,
and $J^n(w)$ refers to $J_2^n(w)$.

Binary constant-weight codes have been extensively studied for more
than four decades due to their fascinating combinatorial structures
and applications
\cite{Johnson:1972,GrahamSloane:1980,Brouweretal:1990,Agrelletal:2000,Barg:2002,Awetal:2003, Zhongetal:2003,EtzionSchwartz:2004,Xiaetal:2005,XingLing:2005,Ji:2006,Kuekesetal:2006,Smithetal:2006,Xiaetal:2006,Chee:2007,FuXia:2007,Gashkovetal:2007,GashkovTaub:2007,CheeLing:2009}. 
Given $n$, $d$, and $w$, the central problem of interest in binary constant-weight codes
is in the determination of $A(n,d,w)$, the largest possible
size of an $(n,d,w)$-code. Exact values
of $A(n,d,w)$ are known only for a few infinite families of parameters $n$, $d$, and $w$,
and in some other sporadic instances (see, for example, \cite{Brouweretal:1990,Agrelletal:2000}).
In light of the difficulty of determining $A(n,d,w)$ exactly,
various bounds have also been developed.
There are two online tables devoted to bounds on
$A(n,d,w)$: one maintained by Rain and Sloane
\cite{RainsSloane} and the other by Smith and Montemanni \cite{SmithMontemanni}. While the
former table considers codes of lengths not exceeding $63,$ the
latter table focuses mainly on codes for lengths between $29$ and
$63$, having small weights.

In this paper, we present simple propagation rules for binary
constant-weight codes through $q$-ary codes.
It turns out that some good binary constant-weight codes can be
obtained from these propagation rules. In particular, we improve on a number of
bounds in the online tables of Rain and Sloane \cite{RainsSloane}, and
Smith and Montemanni \cite{SmithMontemanni}. 

We remark that the table of Smith and Montemanni \cite{SmithMontemanni} was created because the
table of Rains and Sloane \cite{RainsSloane} had not been updated for many years. 
For code parameters that are not covered by Smith and Montemanni \cite{SmithMontemanni},
we have checked against recent literature, to the best of our efforts, in
ascertaining that our results here do indeed improve upon existing results.

\section{Propagation Rules}

In this section, we present some simple propagation rules for
binary constant-weight codes from $q$-ary codes. We begin with
a simple observation.

Let ${\cal C}\subseteq\Z_q^n$. 
For $\vu\in\Z^n_q$, we denote by $\vu+{\cal C}$ the {\em coset} of $\cal C$,
\begin{equation*}
\{\vu+\vc:\;\vc\in {\cal C}\}.
\end{equation*}
We also embed $\Z_2$ into $\Z_q$. It is evident that
$(\vu+{\cal C})\cap J^n(w)$ 
is a binary constant-weight code of weight $w$ and size
$N=|(\vu+{\cal C})\cap J^n(w)|$. Since the minimum distance $d'$ of
$(\vu+{\cal C})\cap J^n(w)$ is at least $d$ and $d'$ must be even, it
follows that 
$d'\geq 2\lfloor(d+1)/2\rfloor$. Thus, we have the following.

\vskip 10pt
\begin{thm}\label{2.1}
Let $0<w<n$.
If there exists an $(n,d)_q$-code $\cal C$, then there exists an
$(n,2\lfloor (d+1)/2\rfloor,w)$-code of size $N$, where
\begin{equation*}
N=\max_{\vu\in\Z_q^n} |(\vu+{\cal C})\cap J^n(w)|.
\end{equation*}
\end{thm}
\vskip 10pt

A simple bound on the size of the constant-weight codes in Theorem \ref{2.1}
can be obtained by considering the average size of
the cosets.

\vskip 10pt
\begin{thm}\label{2.2}
Let $0<w<n$.
If there exists an $(n,d)_q$-code of size $M$,
then
\begin{equation*}
A(n,2\lfloor (d+1)/2\rfloor,w) \geq
\left\lceil\frac{M\binom{n}{w}}{q^n}\right\rceil.
\end{equation*}
\end{thm}

\begin{proof}
Let $\cal C$ be an $(n,d)_q$-code of size $M$.
Let $\vu_1,\vu_2,\ldots,\vu_{q^n}$ denote all the elements of $\Z_q^n$,
and let $\vv_1,\vv_2,\ldots,\vv_{\binom{n}{w}}$ denote all the
elements of $J^n(w)$. Define
\begin{equation*}
\delta_{i,j}=
\begin{cases}
1,&\text{if $\vv_j\in\vu_i+{\cal C}$}\\
0,&\text{if $\vv_j\not\in\vu_i+{\cal C}$}.
\end{cases}
\end{equation*}

For each $\vv_j\in J^n(w)$, there are $M$ elements
$\vu_i\in\Z_q^n$ such that $\vu_i+{\cal C}$ contains $\vv_j$ (to see
this, note that $\vv_j\in\vu_i+{\cal C}$ if and only if
$\vu_i=\vv_j+\vc$ for some $\vc\in {\cal C}$). Thus, 
\begin{equation*}
\sum_{1\leq i\leq q^n} \sum_{1\leq j\leq \binom{n}{w}} \delta_{i,j} = M\binom{n}{w}.
\end{equation*}
Hence, there exists at least one $\ell$,
$1\leq\ell\leq q^n$, such that 
\begin{equation*}
\sum_{1\leq j\leq \binom{n}{w}} \delta_{\ell,j} \geq \frac{M\binom{n}{w}}{q^n}.
\end{equation*}
The theorem now follows by noting
that the size of $(\vu_{\ell}+{\cal C})\cap J^n(w)$ is precisely 
$\sum_{1\leq j\leq \binom{n}{w}} \delta_{\ell,j}$, and we have
seen above that $(\vu_{\ell}+{\cal C})\cap J^n(w)$ is an
$(n,2\lfloor (d+1)/2\rfloor,w)$-code.
\end{proof}
\vskip 10pt

Next, we consider binary constant-weight codes of length $n+1$ from
$q$-ary codes of length $n$.

\vskip 10pt
\begin{thm}\label{2.3}
Let $0<w<n$.
Suppose there exists an $(n,d)_q$-code $\cal C$ of size $M$. Then,
\begin{enumerate}[(i)]
\item
there exists an $(n+1,2\lfloor(d+1)/2\rfloor,w)$-code
of size $N$, where
\begin{equation*}
N=\max_{\vu\in\Z_q^n} |(\vu+{\cal C})\cap (J^n(w-1)\cup J^n(w))|;
\end{equation*}
\item
\begin{equation*}
A(n+1,2\lfloor (d+1)/2\rfloor,w)\geq 
\left\lceil\frac{M(\binom{n}{w-1}+\binom{n}{w})}{q^n}\right\rceil.
\end{equation*}
\end{enumerate}
\end{thm}

\begin{proof} 
\begin{enumerate}[(i)]
\item Let $\vu\in\Z_q^n$ such that 
$|(\vu+{\cal C})\cap (J^n(w-1)\cup J^n(w))|$ 
achieves the maximum size $N$. It is clear that
${\cal C}'=(\vu+{\cal C})\cap (J^n(w-1)\cup J^n(w))$ is an $(n,d)$-code,
where each codeword has weight either $w-1$ or $w$. 
To each codeword $\vc\in{\cal C}'$, 
append a new coordinate which
takes on value one if ${\rm wt}({\vc})=w-1$
and value zero if ${\rm wt}({\vc})=w$. The
set of resulting codewords is an
$(n+1,2\lfloor(d+1)/2\rfloor,w)$-code.

\item Using the same arguments as in the proof of Theorem
\ref{2.2}, we get an $(n,d)$-code of size
$M(\binom{n}{w-1}+\binom{n}{w})/{q^n}$, in which
the weight of every codeword is either $w-1$ or $w$. By
appending a new coordinate to every codeword as in (i) above,
we get an $(n+1,2\lfloor(d+1)/2\rfloor),w)$-code of
the required size.
\end{enumerate}
\end{proof}

\section{Examples}

We provide some examples where the propagation rules
given by Theorems \ref{2.2} and \ref{2.3} lead to improved bounds
on $A(n,d,w)$.

In the tables of this section, a bold entry indicates that
the size of the code constructed here is larger than any known codes of the same parameters,
and a entry superscripted by an asterisk indicates that the size of the code
constructed here is of the same size as the best known code of the same parameters.
$M_{\rm max}$ denotes the lower bound on
$A(n,d,w)$ given by Theorems \ref{2.1} or \ref{2.3}(i),
and $M_{\rm avg}$ denotes the lower bound on $A(n,d,w)$ given by
Theorems \ref{2.2} or \ref{2.3}(ii). $M_{\rm RS}$ denotes the
lower bound on $A(n,d,w)$ in the tables of Rains and Sloane \cite{RainsSloane}.

\vskip 10pt
\begin{example}\label{3.1}
Let $\cal C$ be the Goethals $(63,7)$-code of size $2^{47}$ \cite{Goethals:1974} (see
\cite[Chapter 5]{LingXing:2004} for the structure of this code).

\begin{itemize}
\item Theorems \ref{2.2} and \ref{2.3}(ii) give
\begin{align*}
A(63,8,w) &\geq \left\lceil \binom{63}{w}/2^{16} \right\rceil, \\
A(64,8,w) &\geq \left\lceil\left( \binom{63}{w-1}+\binom{63}{w} \right)/2^{16}\right\rceil.
\end{align*}
The implications of these bounds are given in Table \ref{d=8}.
\begin{table}[h!]
\centering
\caption{Some constant-weight codes of distance eight}
\label{d=8}
\begin{tabular}{r|r|r}
\hline
\multicolumn{3}{r}{Lower Bounds on $A(63,8,w)$} \\
\hline
$w$ & $M_{\rm avg}$ & $M_{\rm RS}$ \\
\hline
7 & {\bf 8443} & 7182 \\
8 & {\bf 59096} & 50274 \\
9 & {\bf 361141} & - \\
10 & {\bf 1950158} & - \\
11 & {\bf 9396214} & -\\
12 & {\bf 40716926} & - \\
13 & {\bf 159735632} & - \\ 
14 & {\bf 570484400} & - \\
\hline
\end{tabular}~~~~~
\begin{tabular}{r|r|r}
\hline
\multicolumn{3}{r}{Lower Bounds on $A(64,8,w)$} \\
\hline
$w$ & $M_{\rm avg}$ & $M_{\rm RS}$ \\
\hline
7 & {\bf 9480} & 8064 \\ 
8 & {\bf 67538} & 57456 \\
9 & {\bf 420236} & - \\
10 & {\bf 2311298} & - \\
11 & {\bf 11346372} & - \\
12 & {\bf 50113140} & - \\
13 & {\bf 200452558} & - \\ 
14 & {\bf 730220032} & - \\
\hline
\end{tabular}
\end{table}

\item Shortening $\cal C$ at the last $i$ positions, 
$1\leq i \leq 46$, results in a $(63-i,7)$-code of size $2^{47-i}$.
It follows from Theorem \ref{2.2} that there exists a
$(63-i,8,7)$-code of size $\binom{63-i}{7}/2^{16}$.
In particular, when $i\in\{1,2,3\}$, this implies
\begin{align}
A(62,8,7) &\geq 7505, \\
A(61,8,7) &\geq 6657, \\
A(60,8,7) &\geq 5894.
\end{align}
The three lower bounds (1)--(3) improve those
in \cite{SmithMontemanni} (the corresponding lower bounds given there are
6693, 6223, and 5770, respectively, obtained by Smith et al. \cite{Smithetal:2006}).
\end{itemize}
\end{example}
\vskip 10pt

\begin{example}\label{3.2}
Let $\cal C$ be the Preparata $(63,5)$-code 
of size $2^{52}$ \cite{Preparata:1968} (see
\cite[Chapter 5]{LingXing:2004} for the structure of this code). 
Theorems \ref{2.2} and \ref{2.3}(ii) give
\begin{align*}
A(63,6,w) &\geq \left\lceil\binom{63}{w}/2^{11}\right\rceil, \\
A(64,6,w) &\geq \left\lceil\left(\binom{63}{w-1}+\binom{63}{w}\right)/2^{11}\right\rceil.
\end{align*}
We also found via computation
cosets of $\cal C$ achieving the maximum in
Theorems \ref{2.1} and \ref{2.3}(i).
The results are given in Tables \ref{n=63,d=6} and \ref{n=64,d=6}. 

\begin{table}[h!]
\centering
\caption{Lower Bounds on $A(63,6,w)$}
\label{n=63,d=6}
\begin{tabular}{r|r|r|r}
\hline
$w$ & $M_{\rm avg}$ & $M_{\rm max}$ & $M_{\rm RS}$ \\
\hline
5 & {3433} & 3906$^*$ & 3906 \\
6 & {33177} & 37758$^*$ & 37758 \\
7 & {\bf 270152} & {\bf 270468} & 264771 \\
8 & {\bf 1891062} & {\bf 1893276} & 1853397 \\
9 & {11556490} & 11594310$^*$ & 11594310 \\
10 & 62405042 & 62609274$^*$ & 62609274 \\
11 & {\bf 300678837} & {\bf 300700062} & 300496392 \\ 
12 & {\bf 1302941625} & {\bf 1302990507} & 1302151032 \\
13 & {5111540218} & {5112164988$^*$} & 5112164988 \\
14 & {18255500778} & 18257732100$^*$ & 18257732100 \\
\hline
\end{tabular}
\end{table}

\begin{table}[h!]
\centering
\caption{Lower Bounds on $A(64,6,w)$}
\label{n=64,d=6}
\begin{tabular}{r|r|r|r}
\hline
$w$ & $M_{\rm avg}$ & $M_{\rm max}$ & $M_{\rm RS}$ \\
\hline
5 & {\bf 3723} & {\bf 3906} & - \\
6 & {36609} & 41664$^*$ & 41664 \\
7 & {\bf 303329} & {\bf 303354} & - \\
8 & {\bf 2161214} & {\bf 2163744} & 2118168 \\
9 & {\bf 13447552} & {\bf 13447707} & - \\
10 & {73961530} & 74203584$^*$ & 74203584 \\
11 & {\bf 363083878} & {\bf 363105666} & - \\
12 & {\bf 1603620460} & {\bf 1603680624} & 1602647424 \\
13 & {\bf 6414481842} & {\bf 6414487191} & - \\ 
14 & {23367040996} & 23369897088$^*$ & 23369897088 \\
\hline
\end{tabular}
\end{table}
\end{example}
\vskip 10pt

\begin{example}
Let $\cal C$ be the (linear) $(31,9)$-code of size $2^{13}$
constructed by Grassl \cite{Grassl}.
\begin{itemize}
\item
We found via computation
cosets of $\cal C$ achieving the maximum in
Theorems \ref{2.1} and \ref{2.3}(i).
The results are given in Table \ref{d=10}.

\begin{table}[h!]
\centering
\caption{Some constant-weight codes of distance $10$}
\label{d=10}
\begin{tabular}{r|r|r}
\hline
\multicolumn{3}{r}{Lower Bounds on $A(31,10,w)$} \\
\hline
$w$ & $M_{\rm max}$ & $M_{\rm RS}$ \\
\hline
11 & {\bf 387} & - \\
12 & {\bf 612} & - \\
13 & {\bf 872} & - \\
14 & {\bf 1106} & - \\ 
\hline
\end{tabular}~~~
\begin{tabular}{r|r|r}
\hline
\multicolumn{3}{r}{Lower Bounds on $A(32,10,w)$} \\
\hline
$w$ & $M_{\rm max}$ & $M_{\rm RS}$ \\
\hline
11 & {\bf 585} & - \\
12 & {\bf 953} & - \\ 
13 & {\bf 1443} & - \\ 
14 & {\bf 1923} & - \\ 
\hline
\end{tabular}
\end{table}

\item Shortening $\cal C$ at the last two positions results in
a (linear) $(29,9)$-code of size $2^{11}$.
We found, via computation, cosets of this shortened code
achieving the maximum in
Theorem \ref{2.3}(i). This gives $A(30,10,12)\geq 390$.
Lower bounds on $A(30,10,12)$ are previously not known.
\end{itemize}
\end{example}
\vskip 10pt

\begin{example}
Let $\cal C$ be the (linear) BCH $(31,11)$-code of size $2^{11}$ 
\cite{Hocquenghem:1959,BoseRay-Chaudhuri:1960}
(see \cite[Chapter 8]{LingXing:2004} for the structure of this code).

\begin{itemize}
\item
We found, via computation,
cosets of $\cal C$ achieving the maximum in
Theorems \ref{2.1} and \ref{2.3}(i).
The results are given in Table \ref{d=12}.

\begin{table}[h!]
\centering
\caption{Some constant-weight codes of distance $12$}
\label{d=12}
\begin{tabular}{r|r|r}
\hline
\multicolumn{3}{r}{Lower Bounds on $A(31,12,w)$} \\
\hline
$w$ & $M_{\rm max}$ & $M_{\rm RS}$ \\
\hline
9 & {\bf 40} & - \\
10 & {\bf 87} & - \\
11 & {\bf 186} & - \\
12 & {\bf 310} & - \\ 
13 & {\bf 400} & - \\
14 & {\bf 510} & - \\
\hline
\end{tabular}~~~
\begin{tabular}{r|r|r}
\hline
\multicolumn{3}{r}{Lower Bounds on $A(32,12,w)$} \\
\hline
$w$ & $M_{\rm max}$ & $M_{\rm RS}$ \\
\hline
9 & {\bf 40} & - \\
10 & {\bf 122} & - \\
11 & {\bf 186} & - \\
12 & {\bf 496} & - \\ 
13 & {\bf 400} & - \\
14 & {\bf 900} & - \\
\hline
\end{tabular}
\end{table}

\item Shortening $\cal C$ at the last $i$ positions, $i\in\{1,2\}$,
results in a $(31-i,11)$-code of size $2^{11-i}$.
We found, via computation, cosets of these shortened codes
achieving the maximum in 
Theorems \ref{2.1} and \ref{2.3} (i). These provide the lower bounds
\begin{align*}
A(29,12,11) &\geq 76, \\
A(29,12,12) &\geq 114,\\
A(29,12,13) &\geq 140,
\end{align*}
and
\begin{align*}
A(30,12,10) &\geq 66, \\
A(30,12,11) &\geq 120, \\
A(30,12,12) &\geq 190, \\
A(30,12,13) &\geq 234, \\
A(30,12,14) &\geq 288.
\end{align*}
Previously, no lower bounds are known on $A(n,12,w)$ for these parameter sets.
\end{itemize}
\end{example}

\vskip 10pt
\begin{example}
Let $\cal C$ be the (linear) $(31,13)$-code of size $2^7$ constructed by Grassl \cite{Grassl}.
We found, via computation, cosets of $\cal C$ achieving the maximum in
Theorem \ref{2.3}(i). These provide the lower bounds
\begin{align*}
A(32,14,12) &\geq 29, \\
A(32,14,13) &\geq 42.
\end{align*}
Lower bounds on $A(32,14,w)$, $w\in\{12,13\}$, are previously not known.
\end{example}

\vskip 10pt
\begin{example}
Let ${\cal C}_0$ be the (linear) Reed-Muller $(32,16)$-code of size $2^6$,
and let $\cal C$ be the
code obtained from ${\cal C}_0$ by puncturing it at the last
position. Then $\cal C$ is a $(31,15)$-code of size $2^6$.
We found, via computation, cosets of $\cal C$ achieving the maximum in
Theorems \ref{2.1} and \ref{2.3}(i). These provide the lower bounds
\begin{align*}
A(n,16,13) &\geq 16, \\
A(n,16,14) &\geq 21, \\
A(n,16,15) &\geq 31,
\end{align*}
for $n\in\{31,32\}$. Lower bounds on $A(n,16,w)$ are previously not known for these
parameters.
\end{example}

\section*{Acknowledgment}
The authors are grateful to Ding Yang and Chen Jie for their help on
programming, to Xing Zhengrong for several
discussions on Theorem \ref{2.1}, and to the anonymous
reviewer for invaluable comments and suggestions.

\end{document}